\documentclass[12pt,reqno]{amsart}
\usepackage{amsmath}
\usepackage{amsthm}
\usepackage{amssymb}
\usepackage[all]{xy}
\usepackage{color}
\usepackage{tikz}
\usepackage{mathrsfs}
\usepackage{orcidlink}
\usepackage{algorithm2e}
\usepackage{float}

\newcounter{obs}

\newtheorem{Theorem}{Theorem}

\newtheorem{Proposition}[Theorem]{Proposition}
\theoremstyle{Definition}

\theoremstyle{remark}

\newtheorem{Example}[Theorem]{Example}

\numberwithin{equation}{section}
\numberwithin{Theorem}{section}

\begin{document}

\title[A Bellman-Ford algorithm for the path-length-weighted distance in graphs]
{A Bellman-Ford algorithm for the path-length-weighted distance in graphs}

\author[Arnau, R.]{R. Arnau\orcidlink{0000-0003-2544-8875}}
\address{%
	Instituto Universitario de Matem\'atica Pura y Aplicada\\
	Universitat Polit\`ecnica de Val\`encia\\
	46022 Valencia\\
	Spain}
\email{ararnnot@posgrado.upv.es (R.A.)}

\author[Calabuig, J.M.]{J. M. Calabuig\orcidlink{0000-0001-8398-8664}}
\address{%
	Instituto Universitario de Matem\'atica Pura y Aplicada\\
	Universitat Polit\`ecnica de Val\`encia\\
	46022 Valencia\\
	Spain}
\email{jmcalabu@mat.upv.es (J.M.C.)}

\author[Garc\'ia Raffi, L.M.]{L. M. Garc\'ia Raffi\orcidlink{0000-0003-3985-8453}}
\address{%
	Instituto Universitario de Matem\'atica Pura y Aplicada\\
	Universitat Polit\`ecnica de Val\`encia\\
	46022 Valencia\\
	Spain}
\email{lmgarcia@mat.upv.es (L.M.G.R.)}

\author[~S\'{a}nchez P\'{e}rez, E.A.]{E.\,A.~S\'{a}nchez P\'{e}rez\orcidlink{0000-0001-8854-3154}}
\address{%
	Instituto Universitario de Matem\'atica Pura y Aplicada\\
	Universitat Polit\`ecnica de Val\`encia\\
	46022 Valencia\\
	Spain}
\email{easancpe@mat.upv.es (E.A.S.P.)}

\author[Sanjuan, S.]{S. Sanjuan\orcidlink{0009-0001-5310-2559}}
\address{%
	Instituto Universitario de Matem\'atica Pura y Aplicada\\
	Universitat Polit\`ecnica de Val\`encia\\
	46022 Valencia\\
	Spain}
\email{ssansil@upvnet.upv.es (S.S.)}



\date{August 21, 2024}

\

\maketitle


\begin{abstract}
Consider a finite directed graph without cycles in which the arrows are weighted. We present an algorithm for the computation of a new distance, called path-length-weighted distance, which has proven useful for graph analysis in the context of fraud detection. The idea is that the new distance explicitly takes into account the size of the paths in the calculations. Thus, although our algorithm is based on arguments similar to those at work for the Bellman-Ford and Dijkstra methods, it is in fact essentially different. We lay out the appropriate framework for its computation, showing the constraints and requirements for its use, along with some illustrative examples.
\end{abstract}


\section{Introduction}

Algorithms for calculating the (weighted)  path-distance between vertices in a graph appeared in the middle of the 20th century, motivated by the growing interest of the time in the applications of mathematical analysis of graphs. The Bellman-Ford algorithm is the main reference of these early studies \cite{bell}. Dijkstra's algorithm for solving the same problem appeared at about the same time \cite{dij}, and differs from the other, being more efficient depending on the particular problem.

After these original works, the growing interest in the subject (due to the numerous applications that graph theory has found in many fields) has given rise to a great deal of research on graph analysis, which often includes the study of these structures when considered as metric spaces. The idea of considering a graph also as a metric space goes back to the beginnings of the theory of graphs. Metric notions begin to appear explicitly in mathematical works in the second half of the last century. The main metric that was considered (and in a sense the only one until the latter part of the century) was the so-called path distance (\cite{entr76,hak,klein02}): for undirected and connected graphs, this metric evaluated between two vertices (nodes) is defined as the length of the shortest path between them (see for example  \cite[\S.2.2.2]{brandes}). One of the first advances in the metric analysis of graphs was the introduction of weights in the definition of the path distance, assigning weights to the individual paths connecting two consecutive nodes and calculating the infimum of the sum of these weights. 
Some recent papers on the subject using weighted distances  that have inspired this paper are  \cite{che11,god11}.

As in the case of other notions of fundamental graph theory, the relevant theoretical ideas appeared together with research topics from other scientific fields, such as sociology \cite{bar83,har53,step}. The definition of different metrics and algorithms to compute them increased greatly in the last decade of the last century, often proposed by problems from other disciplines such as chemistry or crystallography (see \cite{klein02,kle93} and references therein). In this sense, there is a particular case that deserves attention, that is  the  resistance distance \cite{cheb11,yan19}. Coming from some ideas in theoretical chemistry \cite{kle93} and social network analysis \cite{step}, this definition turned out to be a useful tool in the study of molecular configurations in  chemistry, although it has also been used in network analysis and other fields \cite{yang14,bu14,bozzo13}. 
In general,  metrics could play a relevant role  when studying properties such  as robustness  (see for example the survey \cite{oer}; see also \cite{mer}).
The interested reader can find more information on  applications of metric graphs in the books \cite{buck,brandes,fouss}.



In the same direction, in this paper we provide an algorithm to compute a new distance that also appeared in an applied context, in connection with the automatic analysis of fraud in economic networks (see \cite{cal}). In this paper, we show that this metric allows to consider vertices that are far away when the distance is measured using the path-distance, becoming close with respect to this metric. It is especially useful in the economic analysis of fraud in networks of companies, since very often the strategy used to hide such fraud is to use the path-distance. Technically, it is defined as a weighted metric, but by dividing the sums of weights appearing in the infimum that provides the value of the metric by a new term that depends on the number of steps involved in the summation. We will show that this change in the weighting process forces to radically change the algorithm to calculate it, since in this new case, longer paths could give shorter distances.  However, in the present work we consider the case of acyclic directed graphs, in order to avoid restrictions that make it impossible to define a weighted path metric, since continuous passage through a cycle could always give a null value to the metric (which would mean that it is not a metric). There are other ways to avoid this (see for example Proposition 4.1 in \cite{cal}), but in our case we decided to compute the distance between vertices by restricting the set of  possible paths in the infimum that gives it. The way to do this is to avoid cycles, and to consider directed graphs. As a result, what we compute is not a metric on the whole graph, but only a distance between two vertices chosen in it.


Let us give now some basic definitions about graphs and metric spaces. 
Let us explain some concepts related to the general definition of what a metric is, which will adapted later to the graph theoretic framework.
Let $\mathbb R^+$ be the non-negative real numbers.
An (extended) quasi-metric  on a set $\Omega$ is a function $d: \Omega \times \Omega \to \mathbb R^+ \cup \{\infty\}$ such that for all $a,b,c \in \Omega$, the axioms
\begin{enumerate}
	\item $d(a,b)=0=d(b,a)$ if and only if $a=b,$ and
	\item $d(a,b) \leq d(a,c)+d(c,b)$
\end{enumerate}

hold. The resulting quasi-metric structure $(\Omega,d)$ is called a quasi-metric space. 
For the specific framework of this paper, a useful summary of the notions of distance in graphs, with sufficient explanation and many examples, is given in Chapter 15 in \cite{deza}. 

In this paper we will deal with the so called path-length-weighted distance, that was introduced in \cite[\S.4]{cal}.
It should be noted that  the version defined there is given for non-directed graphs, and so the definition we will use is slightly different. However, we will define an extended  quasi-metric also in this case by giving the value $d(a,b) = \infty$ when there is no path for going from $a$ to $b,$ and considering only allowed paths between  vertices. 
Since we are interested in how to compute the distance, and not in theoretical questions about its metric space structure, we will focus attention on the computational algorithm. 

All the notions on graph theory  that are needed can be found in books on this subject, as for example \cite{brandes}. We will introduce some of them in the next section. 

\section{Main definitions and context}

Take a weighted directed graph without cycles $G=(V,E)$.
Let us define what we call a \textit{proximity function.} This is an extended function $\phi: V \times V \to \mathbb R^+ \cup \{\infty\}$ that describes by means of a non-negative real number a
relation among the the node $a$ and the node $b$. It is supposed that  $\phi$ represents some sort of  distance among the nodes, but---and this is crucial---, it is not assumed to be a distance (that is, non of the axioms of metric is assumed). In case  $a$ and $b$ are not connected in the graph, we will write $0, 1$ or $\infty$, depending on how we write the algorithm, that is, how we decide to codify the links in the graph. 
We will write  $(V,E,\phi)$ for the graph when it is provided with a proximity function $\phi.$

Consider a non-increasing sequence $W:=(W_i)_{i=1}^\infty$ of positive real numbers that will be assumed to be the \emph{weights associated to the lengths of the paths} considered in the definition of a distance starting from the proximity function $\phi$. Although all the graphs are finite, we will often define $W$ as an infinite sequence, using only its first elements. Given two vertex $a, b \in V$, we define the set of all possible paths from $a$ to $b$ as
\begin{equation*}
	\mathcal P(a,b) = \{ P = (x_0, x_1, \ldots, x_n) : \ x_0 = a, \ x_n = b, (x_{i-1}, x_i) \in E, \ n \in \mathbb N \}.
\end{equation*}
The \emph{length} of a path $P = (x_0, x_1, \ldots, x_n)$ is denoted by $l(P) = n$, and its \emph{total distance} by
\begin{equation*}
	d(P) = W_n \cdot s(P),
\end{equation*}
where $s(P) = \sum_{i=1}^n \phi(x_{i-1}, x_i)$ is the \emph{sum of path weights} of the path $P.$
Given two points $a,b \in V$, we define
an extended quasi-metric in $V$ as the function on $V \times V$ by
\begin{equation*}
	d_\phi(a,b) = \inf\{ d(P) : \ P \in \mathcal P(a,b) \},
\end{equation*}
with the convention that $\inf \{\emptyset\} =  \infty$.  As in the original definition in \cite[\S.4]{cal}, we will call it the \textit{path-length-weighted quasi-metric. } The reader can find there the proof that it defines a metric. 
Given two paths $P = (x_0, x_1, \ldots, x_n) \in \mathcal P(a,b)$ and $Q = (y_0, y_1, \ldots, y_m) \in \mathcal P(b,c)$ respectively, we define the concatenation of both as 
$$
P \sqcup Q = (x_0, x_1, \ldots x_n, y_1, \ldots, y_m).
$$

Write $V=\{v_i: \, i=0,...,n\}$ for the vertices of the graph. As in the classical math distance case, we are interested in this paper in computing the distance from any of the points in $V$ to $v_0.$


\section{Previous to the algorithm}

Consider the weighted directed graph without cycles $G=(V,E)$ and fix a vertex to which we want to measure the distance from any other vertex of the graph. Write $V=\{v_i: \, i=0,...,n\},$ and suppose that we want to compute the distance from any other node of the graph to $v_0$.

In both the Bellman-Ford and Dijkstra algorithms, the shortest distances from the 
source
node $v_0$ to all other nodes in a graph are computed by selecting the closest unvisited node to the current node. In this process, all edges leaving it to unvisited nodes are examined, updating the distance from these nodes if the distance through the selected node is less than the currently known distance. Thus, the updated distance is given by the path from this node to the current node and the path of least distance between the current node and the source.

We introduce the following formal operation $\boxplus$ to describe the updated distance calculation in a simpler way. Fix a vertex $v_0 \in V,$ and let us write $v_{selected}$ for the vertex from which the distance to $v_0$ we want to compute, and $v_{current}$ another vertex that belongs to one of the paths from $v_{selected}$ to $v_0.$ Let us define 
\begin{equation}\label{eq:sum_dist}
	(d(R), l(R)) \boxplus (d(P), l(P)) := (W_{l(P)+l(R)} \cdot d(R) + \frac{W_{l(P)+l(R)}}{W_{l(P)}} \cdot d(P), \, l(P) + l(R)),
\end{equation}
where $R \in \mathcal{P}(v_{selected},v_{current})$ and $P \in \mathcal{P}(v_{current},v_{0})$ is the path with the shortest distance between the current node and the source.

This weighted sum substitutes the normal addition used in the Bellman-Ford algorithm, and allows to change the weights when the length of the path increases. Unlike the Bellman-Ford algorithm, this algorithm has the additional complexity that the path with the shortest distance between two nodes does not have to be the one that passes through the nodes with the shortest distances, since the number of steps must be taken into account. As illustrated in Example \ref{graph_fail}, replacing the normal sum of the Bellman-Ford algorithm with the weighted sum is not sufficient, as it is crucial to take path length into account in each distance calculation to ensure that no path possibility is discarded.

\begin{Example}\label{graph_fail}
	\begin{figure}[htpb]
		\centering 
		\includegraphics[width=0.8\textwidth]{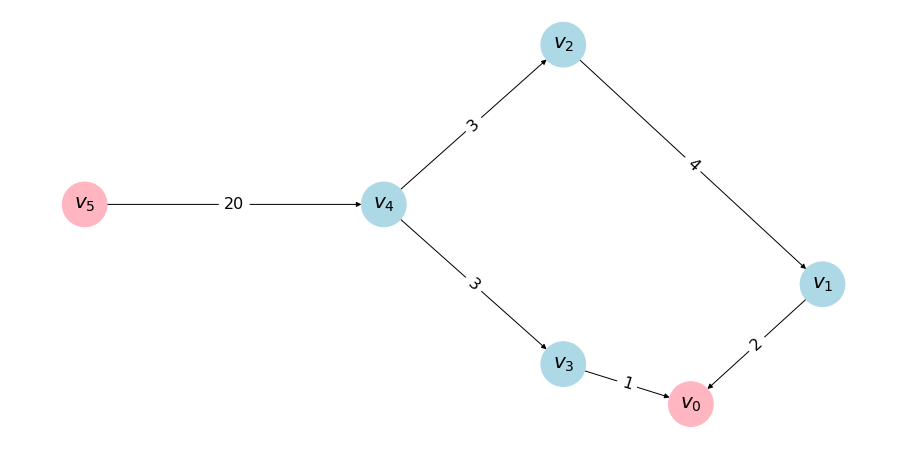}
		\caption{A graph in which a Bellman-Ford type algorithm to calculate the smallest distance between two nodes would not work.  The weight of each edge is written on the corresponding arrow. }
		\label{fig_graph_fail}
	\end{figure}
	
	Consider the graph in Figure \ref{fig_graph_fail}. We will calculate the path-length-weighted  distance from $v_5$ to $v_0$. For this purpose, we consider the weights associated to the lengths of the paths as $W:=\left(\frac{1}{t}\right)_{t=1}^\infty$. This means that the distance between nodes is the average of the path weights.
	
	The first step in calculating the distance has to take into account, of course, the only way forward in the graph, going from $v_5$ to $v_0. $  For the following steps, we clearly have that the path with the smallest distance between the nodes $v_4$ and $v_0$ is 
	$$
	P := (v_4, v_3, v_0),
	$$ 
	with a length $l(P) = 2$ and a distance $d(P) = \frac{3+1}{2} = 2 = d_{\phi}(v_4,v_0).$
	
	Since $v_5$ is only directly connected to $v_4$, applying a Bellman-Ford type algorithm, the computation of $d_{\phi}(v_5,v_0)$ is determined by the path $Q = (v_5,v_4) \sqcup P$, which results in
	$$(d(P), l(Q))= \left(\phi(v_5,v_4), 1\right) \boxplus \left(d(P), l(P)\right) =  \left(\frac{\phi(v_5,v_4)+3+1}{3}, l(P)+1\right) = (8, 3).$$
	In spite of this, the path with the shortest distance between the nodes $v_5$ and $v_0$ is 
	$$
	Q' = (v_5, v_4, v_2, v_1, v_0),
	$$ 
	with a length  $l(Q') = 4$ and an associate  distance  $d(Q') = \frac{20+3+4+2}{4} = 7.25 < 8 = d(Q).$
	
\end{Example}

Consequently, as shown in  Example \ref{graph_fail}, the shortest distance from an unvisited node to the current node for the metric presented in this article does not have to be calculated from the path with the shortest distance between the current node and the source.
Given three vertices $a,b,c \in V$, this implies that the path with the shortest distance from $a$ to $c$ passing through $b$ is not necessarily the union of the paths with the shortest distance from $a$ to $b$ and from $b$ to $c$.
In order not to consider all possible paths in the calculation of distances, multi-objective optimization can be used, which allows to restrict attention to the set of possible paths with the smallest distance, and to make tradeoffs within this set.

Assume that, at a particular step of the algorithm, two or more paths $P, Q \in \mathcal P(a,b)$ are found between the vertices $a, b \in V$.
Define on $\mathcal P(a,b)$ an order given by $P \preceq Q$ if and only if
\begin{equation}\label{eq:orden_relation}
	l(P) \geq l(Q) \ \ \text{and} \ \ s(P) \leq s(Q).
\end{equation}
It is straightforward to see that it is an order relation. 
Note also that it implies $d(P) \leq d(Q)$.
The following result shows that it is sufficient to explore the minimal (in sense of $\preceq$) paths of $\mathcal P(a,b)$, i.e., those residing in the following Pareto front,
\begin{equation}\label{set:filter_path}
	\mathcal P^{*}(a,b) = \left\{ P \in \mathcal P(a,b): \{ P' \in \mathcal P(a,b): P' \preceq P, P \neq P' \} = \emptyset \right\}.
\end{equation}
This front is illustrated in Figure \ref{fig:frontier_pareto}.

\begin{figure}[htbp]
	\centering
	\includegraphics[width=0.9\textwidth]{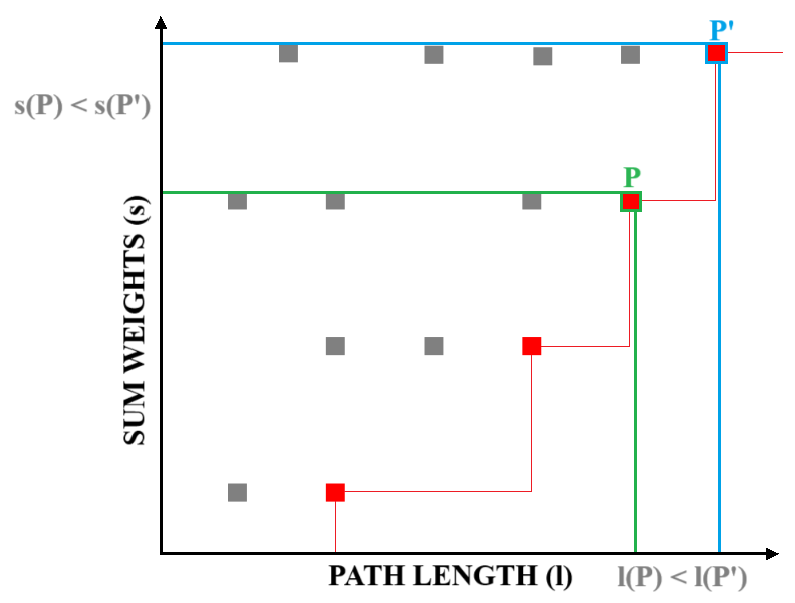}
	\caption{Example of a Pareto front associated to the possible paths in a graph. The framed dots represent all possible paths from one node to another. The red dots represent the set of paths needed to calculate distances, with $\mathcal P^{*}$ being the set of all of them.}
	\label{fig:frontier_pareto}
\end{figure}

\begin{Proposition}\label{prop:filter_path}
	Consider a graph $G$ and the corresponding elements considered above.
	Fix $a, b \in V$ and let $P, P' \in \mathcal P(a,b)$ such that $P \preceq P'$.
	Then, for any $c \in V$ and $R \in \mathcal P(b,c)$, we have that the paths $Q = P \sqcup R$ and $Q' = P' \sqcup R$ in $\mathcal P(a,c)$ satisfies $Q \preceq Q'$ and, in particular, $d(Q) \leq d(Q')$.
	
	This means that, once a better path (in terms of $\preceq$) has been found, the worst one can be discarded (and the exploration stopped) since it will not be used to find any distance between the last vertices.
\end{Proposition}

\begin{proof}
	Consider $P = (x_0, x_1, \ldots, x_n)$, $P' = (x_0', x_1', \ldots, x_m')$ and $R = (y_0, y_1, \ldots, y_r)$ as in the statement.
	Then we have that
	\begin{align*}
		d(Q) & = W_{l(Q)} \cdot \left( \sum_{i=1}^n \phi(x_{i-1}, x_i) + \sum_{i=1}^r \phi(y_{i-1},y_i) \right), \\
		d(Q') & = W_{l(Q')} \cdot \left( \sum_{i=1}^m \phi(x_{i-1}', x_i') + \sum_{i=1}^r \phi(y_{i-1},y_i) \right).
	\end{align*}
	Recall that $P \preceq P'$.
	Then, $l(Q) = l(P) + l(R) \geq l(P') + l(R) = l(Q'),$ and  so $W_{l(Q)} \leq W_{l(Q')}$.
	Moreover,
	\begin{equation*}
		\sum_{i=1}^n \phi(x_{i-1}, x_i) = s(P) \leq s(P') = \sum_{i=1}^m \phi(x_{i-1}', x_i').
	\end{equation*}
	Thus, we get $Q \preceq Q'$.
\end{proof}

\textit{Remark.}
It is  easy to see that Proposition \ref{prop:filter_path} also works if the composition of the  paths is done in the opposite way. That is, the same conclusion  holds if  $Q = R \sqcup P$ and $Q' = R \sqcup P'$ belong to $\mathcal P(a,c),$ with $R \in \mathcal P(a,b)$ and $P, P' \in \mathcal P(b,c)$.

\section{The algorithm}\label{The algorithm}

Taking into account the ideas of the previous section, in this part of the paper we explain the algorithm to compute the path-length-weighted distance in a graph.
To make the new algorithm more efficient, the sum of the weights ---$s(P)$--- will be used  in the calculations instead of the distance ---$d(P)$---, which will be calculated in the last step.
The counting parameter $m$ will be used to indicate the step number in the algorithm, that is, the length of the path.
Let $V=\{v_i: \, i=0,...,n\}$ denote the set of all nodes, where $v_0$ represents the 
source node.
In this context, we define $D_i$, with $v_i \in V$, as the set of all computed tuples as $(s_i,l_i)$, where (to simplify notation) $s_i$ denotes the sum of all weights of a path $P \in \mathcal P(v_i,v_0),$ and $l_i$ represents its length.
We introduce the following formal operation $\oplus$ to describe the process of adding a new node  in terms of the sum of all weights. Once a couple of adjacent  vertices $v_i,v_j$ are fixed, we write 
\begin{equation}\label{eq:adding_nodes}
	\phi(v_j,v_i) \oplus (s_i, l_i) := \left(s_i + \phi(v_j,v_i), l_i + 1\right).
\end{equation}
This definition is closely related to (but of course not the same as) that of $\boxplus$ given in equation \ref{eq:sum_dist}, but, as we noted above, it is written in terms of $s$ rather than $d.$
Finally, we set the (ordered) set $d_{\phi} = \{ d_{\phi}(v_i,v_0): v_i \in V \}.$ 

The new algorithm that we propose follows the next steps. The parameter $m \in \mathbb N$ gives the size of the paths that are considered at each step connecting any vertex $v_i$ with the source node $v_0.$ 

\begin{itemize} 
	
	\item[(1)] \textbf{Initialization}. For $m=0$, initialize the algorithm by doing $D_i = \{(\infty, 0)\}$ for $i=1,...,n-1,$ and $D_0 = \{(0, 0)\}$.
	We have that $l_i=0$ for all $i$, since we are considering steps of length equal to $0$ at the beginning.
	Now, consider the ``set of vertices connected to $v_0$ by a path of length $0$", that is, $\{v_0\}.$
	Put also, according to the definition,    $s_0=l_0=0.$ 
	
	\item[(2)] \textbf{Search for all paths}. Put $m=1.$ 
	Consider the set of adjacent vertices to $v_0,$
	$$
	A_{0}:= \{v_j \in V  \setminus \{v_0\} : \, \phi(v_j,v_0) \in \mathbb R^+ \},
	$$ 
	that is, all the vertices that are connected to $v_0$. For  $v_j \in A_{0},$ consider the quantities
	$$
	\phi(v_j,v_0) \oplus (s_0, l_0) = (\phi(v_j,v_0), 1)
	$$
	given by equation \ref{eq:adding_nodes} and add them to the set $D_j$ that has been initialized as only containing $(\infty, 0),$  that is,
	$$
	D_j = D_j \cup \{ (\phi(v_j,v_0), 1)\} 
	.$$
	
	\item[(3)] \textbf{Path filtering}. The main idea of the algorithm is to update the sets $D_j$ as $m$ grows when repeating step (2). For a given $m$ following $1,$ (make $m=2$ for the immediate step) we have  as $D_j$ the set of tuples $(s,l)$ of all paths $P \in \mathcal P(v_j,v_0)$ of length less than or equal to $m.$
	We need to reduce the size of  these sets by eliminating some non-essential elements. 
	By Proposition \ref{prop:filter_path}, the equation \ref{set:filter_path} will be used to discard the tuples belonging to those paths who  give a greater distance than others. We do that by replacing
	$$D_j \quad \text{by} \quad  \mathcal P^{*}(D_j),$$
	that is defined for any $D_j$ 
	as
	\begin{equation*}\label{eq:path_filtering}
		\qquad	\mathcal P^{*}(D_j) = \big\{ (s,l) \in D_j: \{ (s',l')  \in D_j:  \, \, (s',l') \ne (s,l), \,\,\, s' \leq s \text{ and } l' \geq l \} = \emptyset \big\}.
	\end{equation*}
	
	Note that in the case that a given $v_j$ has any path of size $m$ to go to $v_0$, this step will remove the tuple $(\infty,0)$ from the set $D_j.$
	
	%
	
	\item[(4)] 
	\textbf{Repeat}. Repeat steps $(2)$ and $(3)$ until $m$ is greater than the maximum length of all paths in the graph  connecting $v_0.$
	
	%
	%
	
	%
	%

	\item[(5)] \textbf{Distance computation}. Finally, once the final versions of the sets $D_i$ have been calculated, we proceed to the calculation of the distances. Then, for each $v_i \in V$, the distance from the node $v_i$ to the 
	source 
	node $v_0$ is given by
	$$
	d_{\phi}(v_i, v_0) = \min_{(s_i,l_i) \in D_i}\left\{W_{l_i} \cdot s_i\right\}.
	$$
	Note that this way of calculating the distance, although it can be laborious, guarantees that it works for any non-increasing sequence of weights $W.$ 
	
\end{itemize}

The reader can find below Algorithm \ref{alg:path_length_weight}, a visual representation illustrating the steps of the proposed algorithm. This diagram serves as a helpful visual aid to complement the step-by-step explanation provided above.

\begin{algorithm}\label{alg:path_length_weight}
	\caption{path-length-weight($\phi$, $W$, $n$, $v_0$)}
	\label{alg:path_length_weight}
	\SetKwInOut{Input}{Input}
	\Input{$\phi$, $n$, $W$, $v_0$}
	$V \leftarrow \{v_i: \, i \in \{0, \dots n\}\}$\;
	$D \leftarrow \{ D[v_i] = \{[\text{inf}, 0]\}: \, i \in \{1, \dots n\} \}$\;
	$D[v_0] \leftarrow \{[0,0]\}$\;
	$Q \leftarrow \{v_0\}$\;
	$m \leftarrow 1$\;
	\While{$(m \leq n) \text{ and } (Q \neq \emptyset)$}{
		\For{$v_i$ \text{ in } $Q$}{
			$A_i \leftarrow \{v_j \in V - \{v_i\}: \, \phi[v_j,v_i] \neq \text{inf}\}$\;
			\For{$v_j$ \text{ in } $A_i$}{
				\For{$[s_i,l_i]$ \text{ in } $D[v_i]$}{
					$s \leftarrow s_i + \phi[v_j,v_i]$\;
					\If{$[s, l_i+1] \notin D[v_j]$}{
						Add $[s, l_i+1]$ to $D[v_j]$\;
					}
				}
				$D[v_j] \leftarrow \mathcal P^{*}(D[v_j])$\;
			}
		}
		$Q \leftarrow \{v_i \in V - Q: \, [\text{inf}, 0] \notin D[v_i]\}$\;
		$m \leftarrow m + 1$\;
	}
	$d \leftarrow (0,...,n)$ 
	
	\For{$v_i$ \text{ in } $V$}{
		$d[v_i] \leftarrow \min\left(W[l_i] \cdot s_i: \, [s_i,l_i] \in D[v_i]\right)$
	}
	\Return $d$\;
	
\end{algorithm}

\section{Examples of graphs}

In this section we will present some examples to show the properties of the explained algorithm. We will put some different cases in order to demonstrate that the conditions of the requirements are necessary to get a suitable algorithm. Also, our idea is to present different classes of graphs, which can be associated using topological properties, to show that our technique works in different situations.

We start with an example of a tree type graph. After that, we will show that our algorithm also works for a centered-star type graph. Finally, we will also present a case of a graph with no singular topological structure to prove the generality of our procedure.

To visualize the graphs, the Kamada-Kawai layout algorithm \cite{kam89} is used, a method for drawing two-dimensional graphs where the nodes are placed so that the distances in the drawing are proportional to the shortest distances between them in the original graph. In this way it is possible to observe how the drawing of the graph changes depending on the distance used. The numbers shown on the arrows are the corresponding values of $\phi,$ except where other notation is explained. Let us remark that the ``distances"  that are written in the figures are the ones that can be computed: recall again that we are working with directed graphs with no cycles, so it may happen that we can compute the ``distance" from $v_1$ to $v_2$ but not the ``distance" from $v_2$ to $v_1.$ In this case, we write in the corresponding edge the one that can be computed. 

\begin{itemize}
	\item[(1)] \textbf{Tree type graph}. In this type of graphs all nodes are connected to each other by exactly one single path, and there are no cycles (closed paths). For the calculation of the distance we  use the weights associated with the inverse of the lengths of the paths ---this is $W:=\left(\frac{1}{t}\right)_{t=1}^\infty$---. 
	The comparison among Figure \ref{fig_graph_tree} and Figure \ref{fig_graph_tree_distances}
	shows the particular nature of the metric we have defined, in which two points $v_0$ and $v_i$ that are more distanced in the tree (with more intermediate nodes) are however  nearer than one of the vertex that crosses the path that connect $v_0$ with $v_i.$ The reader can see this, for example, in the comparison of the position and distances between $v_0$ and $v_{11}$ (the path-length-weighted distance is $2.333$), and between $v_0$ and $v_{13}$ (the path-length-weighted distance is $2.250$). Moreover, the distances between $v_0$ and $v_1,$ and between $v_0$ and $v_3,$ are equal, although it is necessary to pass through $v_1$ to reach $v_0$ from $v_3.$ This effect is especially remarkable in the case of trees, where the usual path distance between a vertex and the root increases with the number of branches separating them. We have seen that for our distance, however, this is not necessarily true. 
	
	\begin{figure}[H]
		\centering 
		\includegraphics[width=0.63\textwidth]{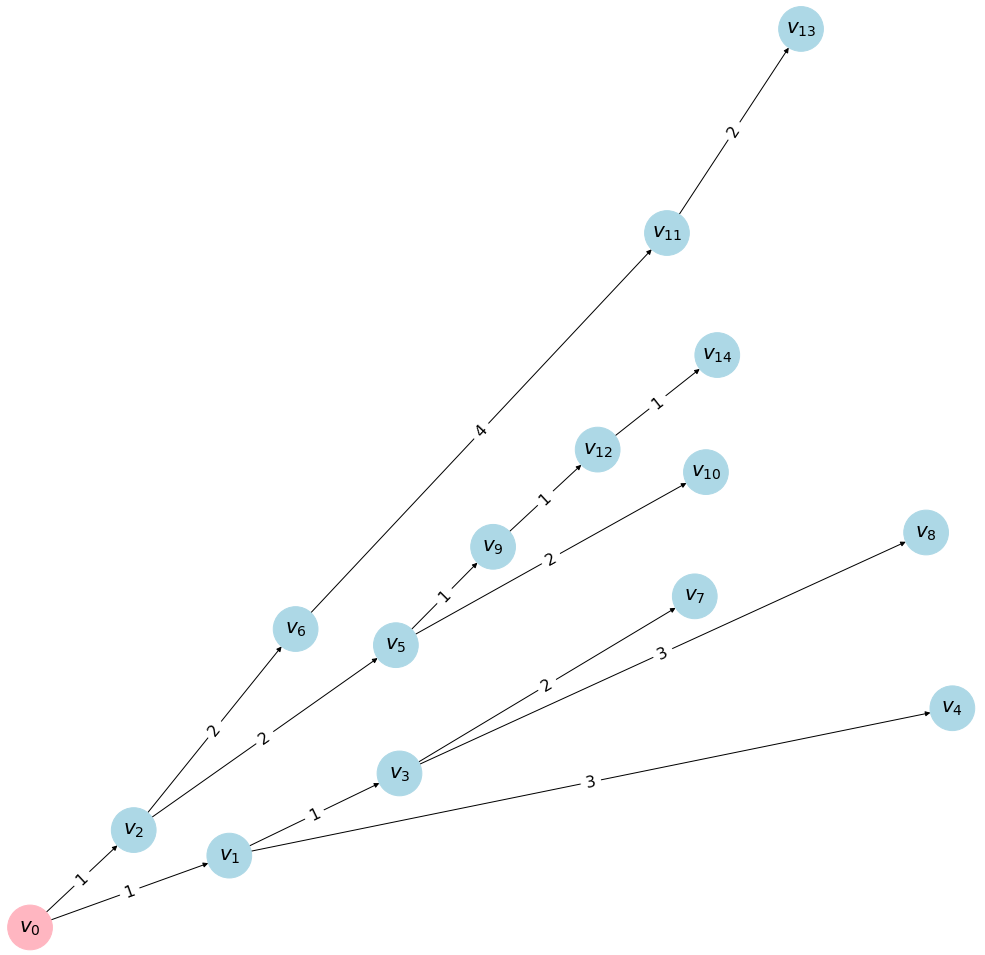}
		\caption{Example of a tree type graph.}
		\label{fig_graph_tree}
	\end{figure}
	
	\begin{figure}[H]
		\centering 
		\includegraphics[width=0.63\textwidth]{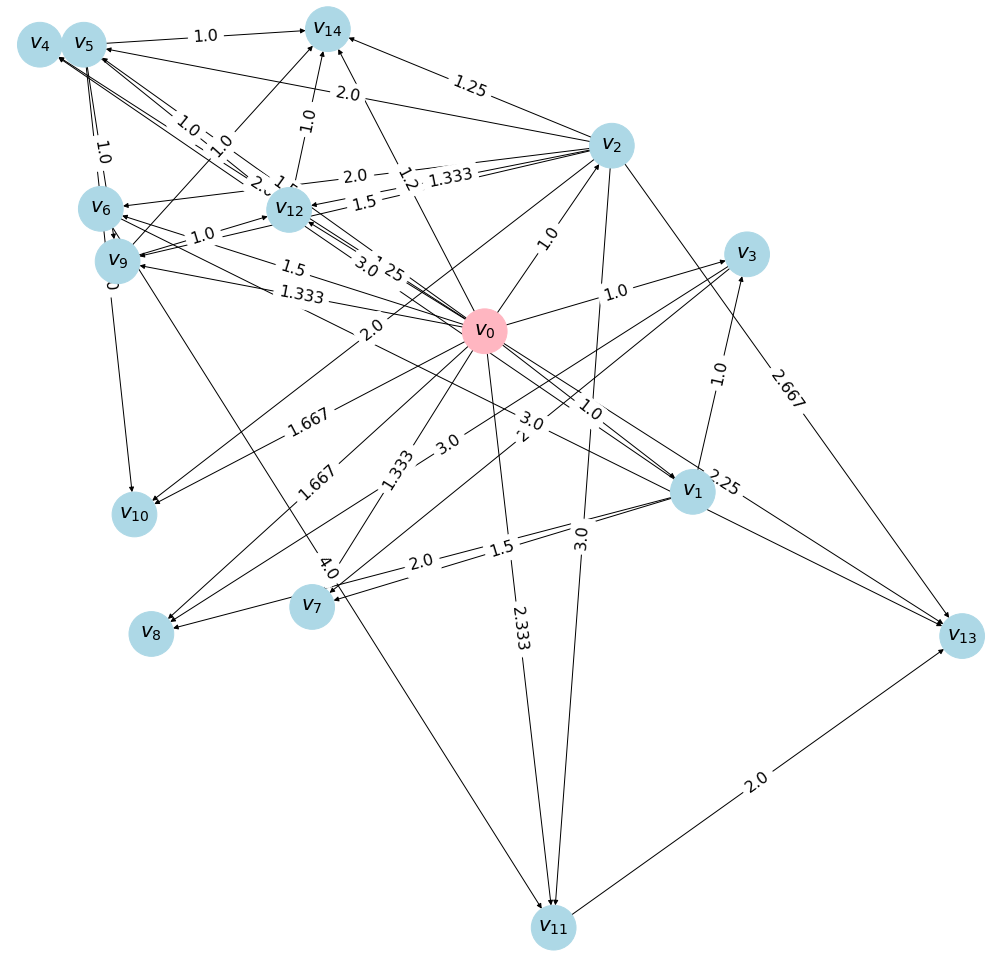}
		\caption{Distances calculated from the graph in the Figure \ref{fig_graph_tree}. In this case, the numbers appearing in the arrows are the quasi-metrics, and not the values of $\phi.$}
		\label{fig_graph_tree_distances}
	\end{figure}

	\item[(2)] \textbf{Centered-star type graph}. In these graphs there is a central node ---$v_0$ in Figure \ref{fig_graph_centered}--- which is directly connected to all other nodes in the graph.
	
	\begin{figure}[H]
		\centering 
		\includegraphics[width=0.63\textwidth]{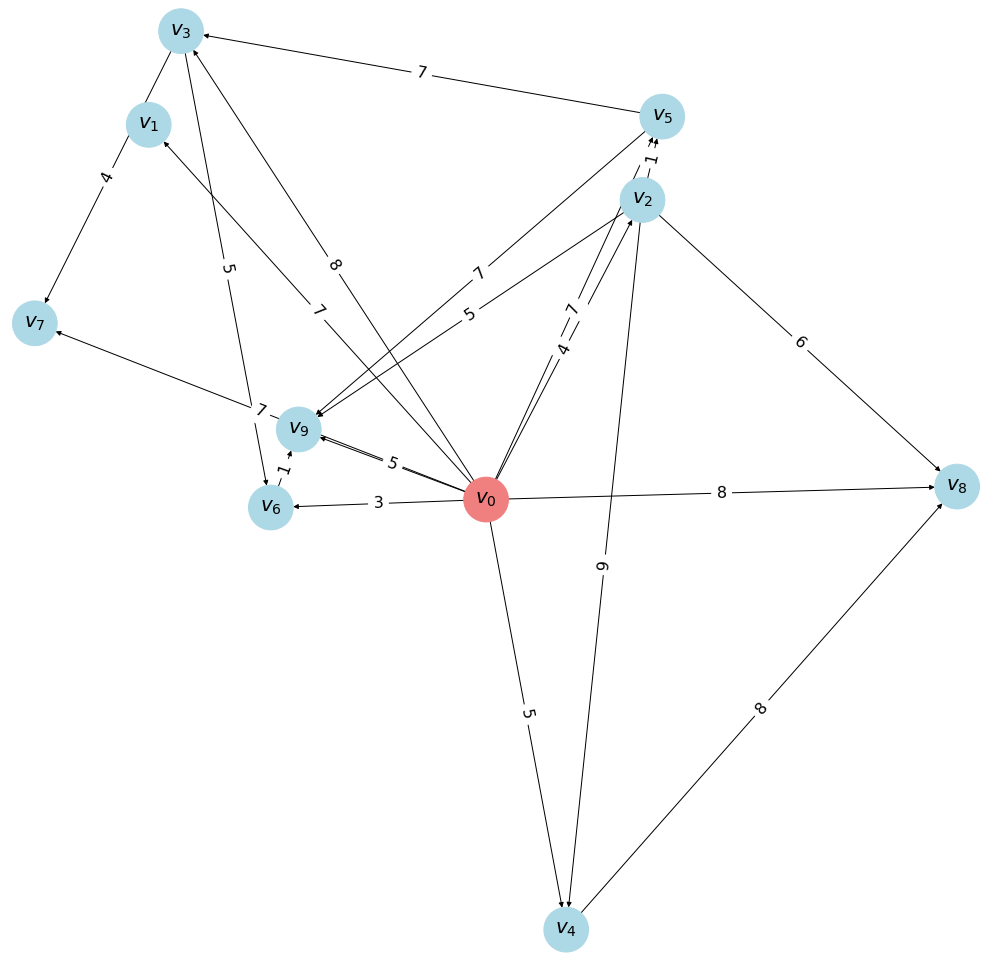}
		\caption{Example of a centered-star type graph.}
		\label{fig_graph_centered}
	\end{figure}
	
	\begin{figure}[H]
		\centering 
		\includegraphics[width=0.63\textwidth]{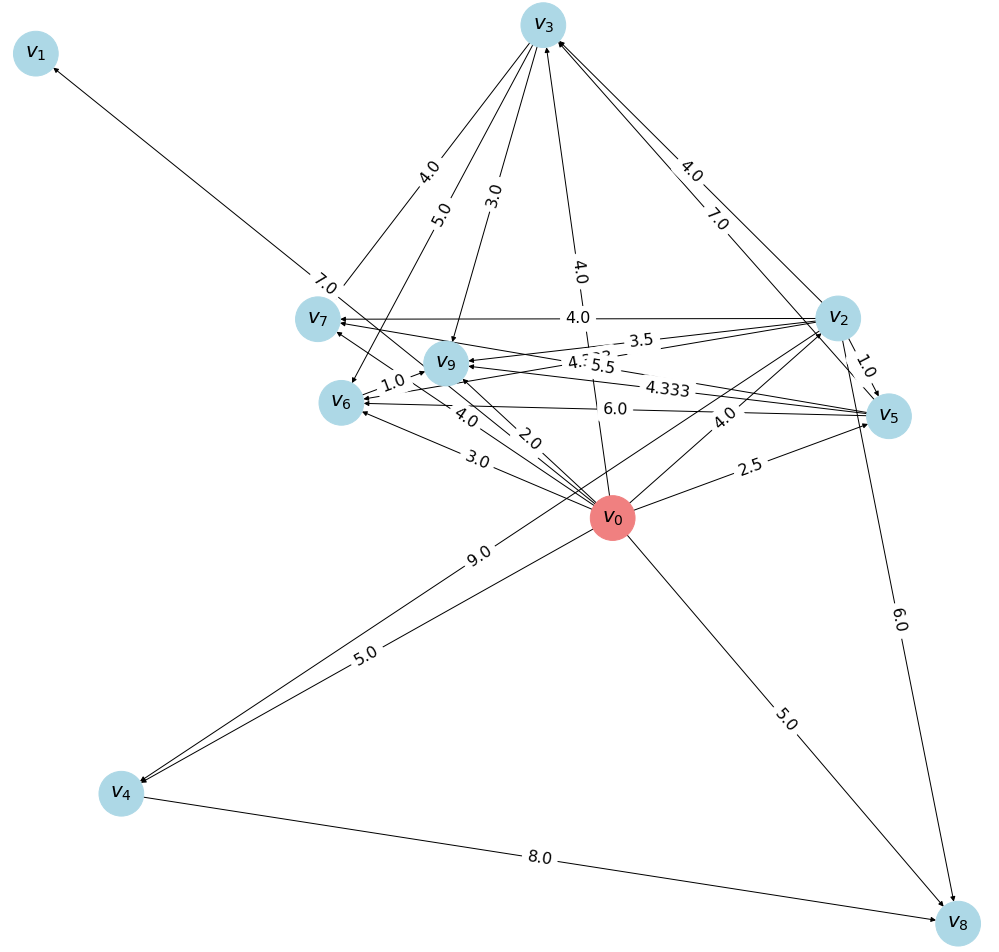}
		\caption{Distances calculated for the graph in the Figure \ref{fig_graph_centered} for $W=\left(\frac{1}{t}\right)_{t=1}^\infty.$}
		\label{fig_graph_centered_distances}
	\end{figure}
	
	Figure \ref{fig_graph_centered_distances_weights2} shows how the graph changes when applying the distance calculation with the weights $W=\left(\frac{1}{t^2}\right)_{t=1}^\infty$  compared to Figure \ref{fig_graph_centered_distances}, where the weights  $W=\left(\frac{1}{t}\right)_{t=1}^\infty$ are used. In this case the distance decreases as the path length increases. This is again a relevant difference with the usual weighted path distance case, and could be used to model different network behaviors where proximity in terms of number of intermediate nodes does not adequately represent the desired relationships between nodes.
	\begin{figure}[H]
		\centering 
		\includegraphics[width=0.63\textwidth]{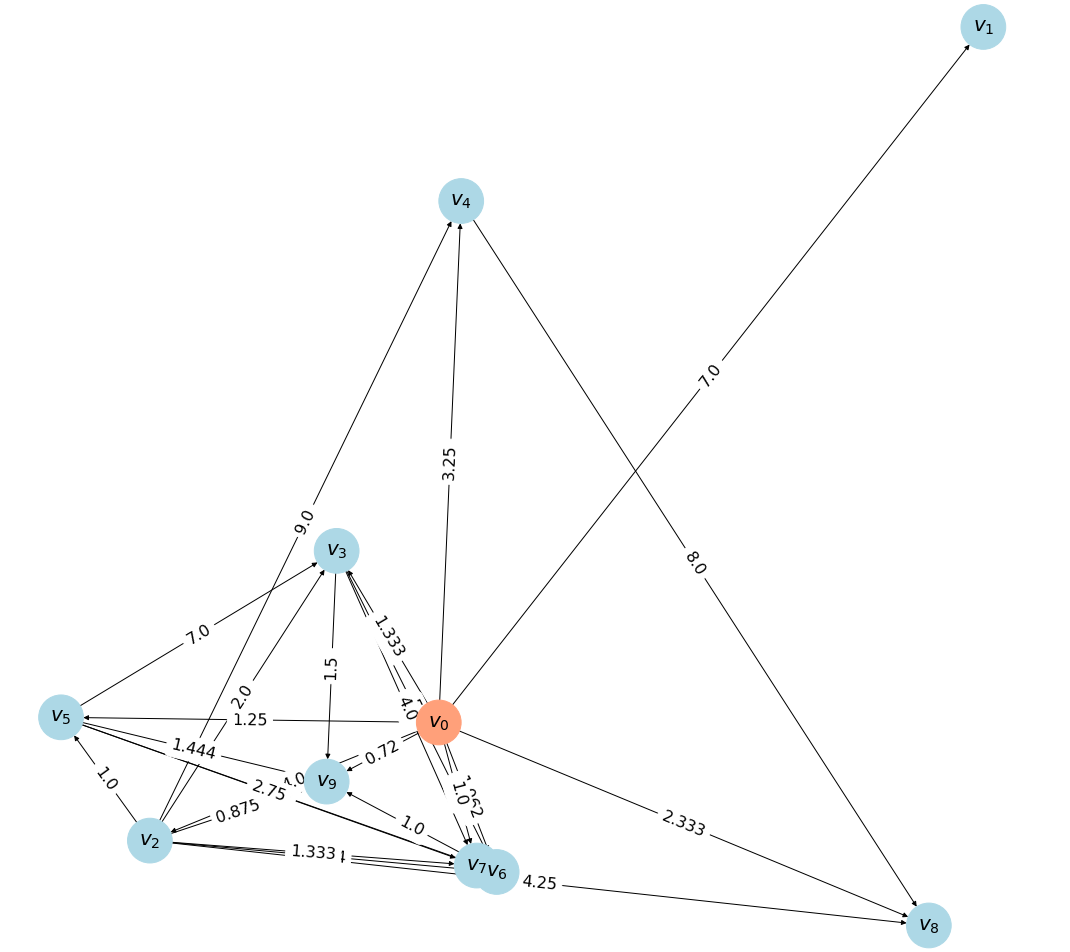}
		\caption{Distances calculated for the graph in the Figure \ref{fig_graph_centered} using the square of the inverse of the path lengths $W=\left(\frac{1}{t^2}\right)_{t=1}^\infty.$}
		\label{fig_graph_centered_distances_weights2}
	\end{figure}
	
	\item[(3)] \textbf{Graph with no singular topological structure}. Let us now consider a graph that does not have an easily definable shape or pattern, such as a tree or a star. 
	The example that we choose is shown in Figure \ref{fig_graph_separated}. It  is a non-connected graph.
	\begin{figure}[H]
		\centering 
		\includegraphics[width=0.63\textwidth]{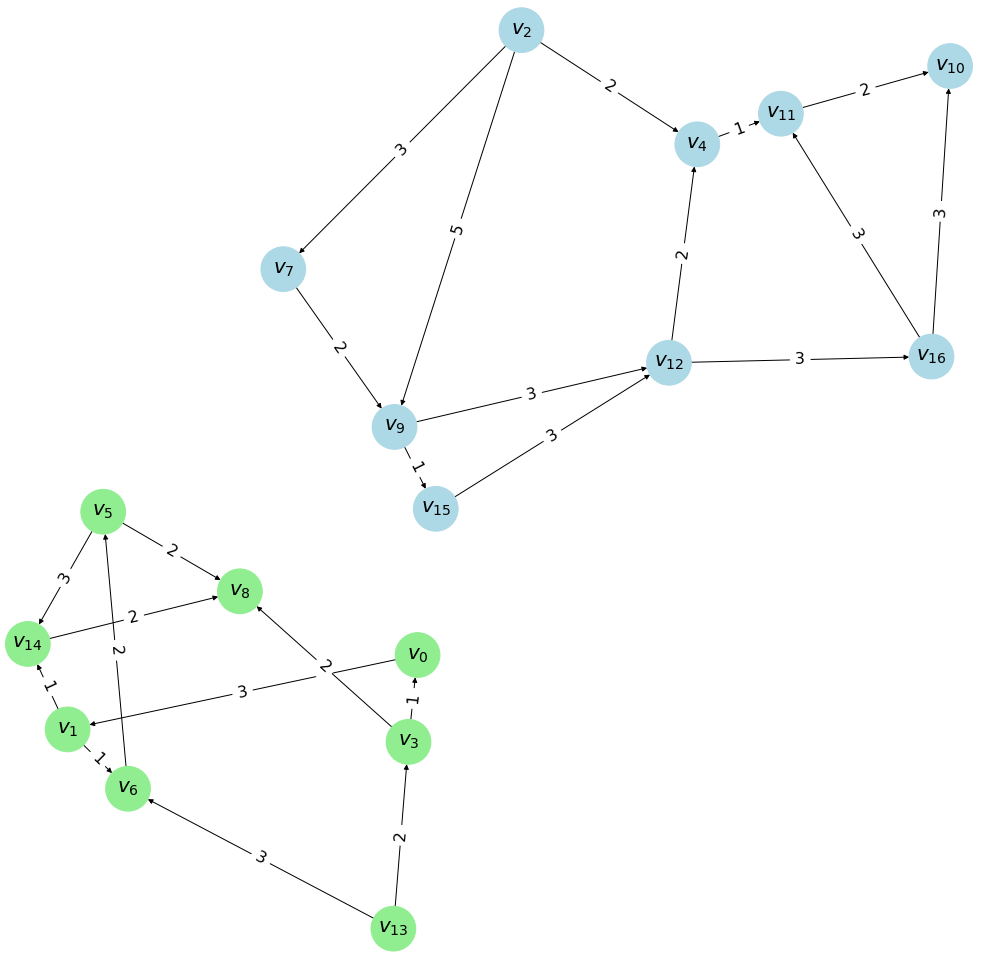}
		\caption{Example of a graph with no singular topological structure.}
		\label{fig_graph_separated}
	\end{figure}
	
	This graph can be separated into two connected subgraphs ---that is, its two connected components--- to calculate the distances of the nodes. In spite of this, as can be seen in Figure \ref{fig_graph_separated_distances}, our algorithm also works without the need to separate these graphs into their connected components. The weights $W=\left(\frac{1}{t}\right)_{t=1}^\infty$ are again considered.
	
	\begin{figure}[H]
		\centering 
		\includegraphics[width=0.63\textwidth]{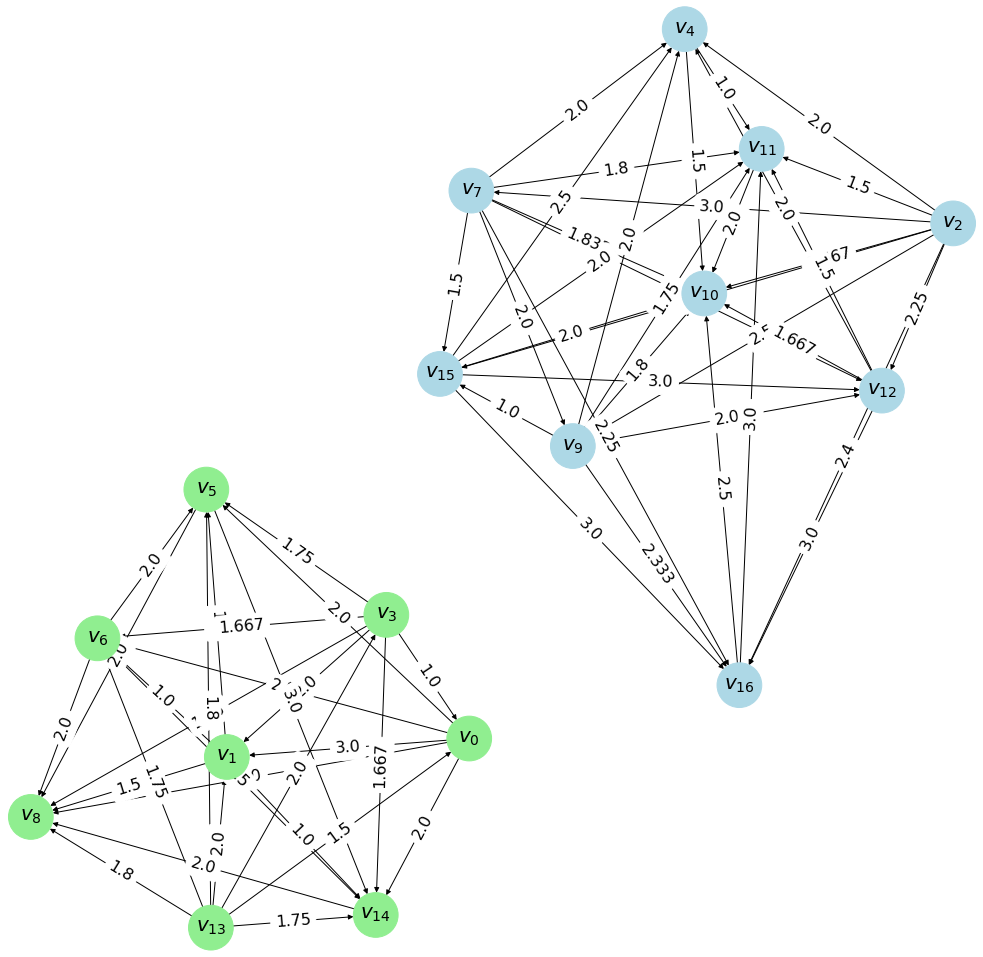}
		\caption{Distances calculated for the graph in Figure \ref{fig_graph_separated}.}
		\label{fig_graph_separated_distances}
	\end{figure}
	
\end{itemize}

\section{Special cases: constraints on weights and proximity matrix to improve the efficiency of the algorithm}

To reduce the complexity of the algorithm, algebraic conditions can be imposed on both the proximity matrix and the weights associated with the length of the paths.
The trivial case is when the weights associated with the path lengths are constant, $W.$ In this case the distance coincides with the weighted path metric (see \cite{deza}, p.258) multiplied by the constant $W.$ Therefore, the calculation of the distance can be performed using the Bellman-Ford or Dijkstra algorithms.

In this section we will consider the weights associated with the inverse of a power of the lengths of the paths $W:=\left\{\frac{1}{t^k}\right\}_{t=1}^\infty$, with the power $k$ greater than or equal to $1$. The importance of these weights lies in the fact that it is the natural generalization of the  case when the power is equal to $1$, the distance between nodes coincides with the average of the path weights. This class of weights, that have been already used in the previous sections as general examples, is also relevant for applications (see \cite{cal}).

The results provided in this section can be implemented directly in the algorithm to make it faster. We explain the requirements that have to be met, as well as the procedures for using them. 

\begin{Proposition}\label{prop:Pareto2}
	Consider a graph $G = (V,E)$ and the weights associated to the lengths of the paths $W:=\left\{\frac{1}{t^k}\right\}_{t=1}^\infty$, with $k \geq 1$.
	Fix $a, b \in V$ and let $P, P' \in \mathcal P(a,b)$ such that $l(P) \geq l(P')$ and $d(P) \leq d(P')$.
	Then, for any $c \in V$ and $R \in \mathcal P(b,c)$ such that $d(P') \leq d(R)$, we have that the paths $Q = P \sqcup R$ and $Q' = P' \sqcup R$ in $\mathcal P(a,c)$ satisfy $d(Q) \leq d(Q')$.
\end{Proposition}

\begin{proof}
	Consider $P = (x_0, x_1, \ldots, x_n)$, $P' = (x_0', x_1', \ldots, x_m')$ and $R = (y_0, y_1, \ldots, y_r)$ as in the statement so that $l(P) = n$, $l(P') = m$ and $l(R) = r$.
	Then,
	\begin{align*}
		d(Q) &= d(P \sqcup R) = \frac{\sum_{i=1}^n \phi(x_{i-1}, x_i) + \sum_{i=1}^r \phi(y_{i-1},y_i)}{\left(n + r\right)^k} = \frac{n \cdot d(P) + r \cdot d(R)}{\left(n + r\right)^k}, \\
		d(Q') &= \frac{m \cdot d(P') + r \cdot d(R)}{\left(m + r\right)^k}.
	\end{align*}
	Then by hypothesis,
	\begin{align*}
		d(P) &\leq \frac{m \cdot d(P')}{m} = \\
		&= \frac{m}{n \cdot m} \left(m \cdot d(P') \right) + \frac{n-m}{n \cdot m} \left(m \cdot d(P') \right) \leq \\
		&\leq \frac{m \cdot d(P')}{n} + \frac{n-m}{n} d(R).
	\end{align*}
	That implies that $n \cdot d(P) \leq m \cdot d(P') + (n-m) d(R)$. Therefore,
	\begin{align*}
		\left(n + r\right)^k (m+r)d(Q) &= \left(n + r\right)^k (m+r) \left( \frac{n \cdot d(P) + r \cdot d(R)}{\left(n + r\right)^k} \right) = \\
		&= n \cdot m \cdot d(P) + n \cdot r \cdot d(P) + (m+r) (r \cdot d(R)) \leq \\
		&\leq n \cdot m \cdot d(P') + r \cdot m \cdot d(P') + (n-m) (r \cdot d(R)) + (m+r) (r \cdot d(R)) = \\
		&= (n+r) \left( m \cdot d(P') + r \cdot d(R) \right) = \\
		&= (n+r)\left(m + r\right)^k d(Q') \leq \\
		&\leq  \left( \frac{n+r}{m+r} \right)^{k-1} (n+r)\left(m + r\right)^k d(Q') = \\
		&= \left(n + r\right)^k (m+r) d(Q'),
	\end{align*}
	so $d(Q) \leq d(Q')$.
\end{proof}

\begin{Proposition}\label{prop:Pareto3}
	Consider a graph $G = (V,E)$ and the weights associated to the lengths of the paths $W:=\left\{\frac{1}{t^k}\right\}_{t=1}^\infty$, with $k \geq 1$.
	Fix $a, b \in V$ and let $P, P' \in \mathcal P(a,b)$ such that $l(P) \leq l(P')$ and $d(P) \leq d(P')$.
	Then, for any $c \in V$ and $R \in \mathcal P(b,c)$ such that $d(R) \leq d(P)$, we have that the paths $Q = P \sqcup R$ and $Q' = P' \sqcup R$ in $\mathcal P(a,c)$ satisfy $d(Q) \leq d(Q')$.
\end{Proposition}

The proof of this result is analogous to Proposition \ref{prop:Pareto2}, so it will be omitted.

\textit{Remark.}
It is trivial that Propositions \ref{prop:Pareto2} and \ref{prop:Pareto3} also work if the paths $Q = R \sqcup P$ and $Q' = R \sqcup P'$ in $\mathcal P(a,c)$ are used, being in this case $R \in \mathcal P(a,b)$ and $P, P' \in \mathcal P(b,c)$.

These propositions suggest to define in $\mathcal P(a,b)$ and for each $R \in \mathcal P(b,c)$ the order relation $P \preceq_{R}^1 Q$  given by the conditions
\begin{equation*}
	l(P) \geq l(Q) \ \ \text{and} \ \ d(P) \leq d(Q) \leq d(R)
\end{equation*}
(suggested by Proposition \ref{prop:Pareto2}), and the order relation $P \preceq_{R}^2 Q$ given by
\begin{equation*}
	l(P) \leq l(Q) \ \ \text{and} \ \ d(R) \leq d(P) \leq d(Q)
\end{equation*}
(suggested by Proposition \ref{prop:Pareto3}).

As can be seen, these order relations are more restrictive than the order relation $\preceq$ defined in the previous sections (although they depend on the path $R$), since they are given as a function of the distance $d$ while $\preceq$ is a function of the sum of the weights $s$.

To reduce the set of paths suitable for calculating the distance with our algorithm using these order relations we need them to not depend on the path $R$. This can be done by imposing some additional properties to the graph. In particular, this works if  the graph meets one of the following conditions, that are different depending on the order relation we will use. 

\begin{itemize}
	\item[1] \textit{Requirements for the order $\preceq_{R}^1$.}  Let $G = (V,E)$ be a weighted directed graph with no cycles. The condition we impose in this case is the following: for all $a,b,c$ in $V$, with $\mathcal P(a,b)$ nonempty and $b$ connected to $c$ ---that is $(b,c) \in E$---, the inequality
	\begin{equation}\label{eq:condition2}
		d(P) \leq \phi(b,c)
	\end{equation}
	holds for all $P \in \mathcal P(a,b)$.  Note that this requirement is independent of $R$, so we write 
	$\preceq^1$ for the associate order relation.
	As it is not always possible to verify this condition given a graph, we will check that for all $(a,b), (b,c) \in E$ it is satisfied that
	\begin{equation}\label{eq:condition2_2}
		\phi(a,b) \leq \phi(b,c).
	\end{equation}
	It is obvious that  this condition implies \ref{eq:condition2}.
	Using this, we can simplify the algorithm. Clearly, in this case it is enough to explore the paths of the following Pareto front
	\begin{equation}\label{set:filter_path2}
		\mathcal P_{1}^{*}(a,b) = \left\{ P \in \mathcal P(a,b): \{ P' \in \mathcal P(a,b): P' \preceq^1 P, P \neq P' \} = \emptyset \right\}.
	\end{equation}

	\item[2] \textit{Requirements for the order $\preceq_{R}^2.$} As above, let $G = (V,E)$ be a weighted directed graph with no cycles. The requirement is in this case the following: for all $a,b,c$ in $V$, with $\mathcal P(a,b)$ nonempty and $b$ connected to $c$ ---that is $(b,c) \in E$---,  the inequality
	\begin{equation}\label{eq:condition3}
		d(P) \geq \phi(b,c)
	\end{equation}
	has to be satisfied for all $P \in \mathcal P(a,b)$. 
	A sufficient condition to be checked is that for all $(a,b), (b,c) \in E$ it is satisfied that $\phi(a,b) \geq \phi(b,c).$ Again, this implies  \ref{eq:condition3}.
	In this case it is enough to consider the  paths: 
	\begin{equation}\label{set:filter_path3}
		\mathcal P_{2}^{*}(a,b) = \left\{ P \in \mathcal P(a,b): \{ P' \in \mathcal P(a,b): P' \preceq^2 P, P \neq P' \} = \emptyset \right\}.
	\end{equation}
\end{itemize}


\begin{figure}[htpb]
	\centering 
	\includegraphics[width=0.8\textwidth]{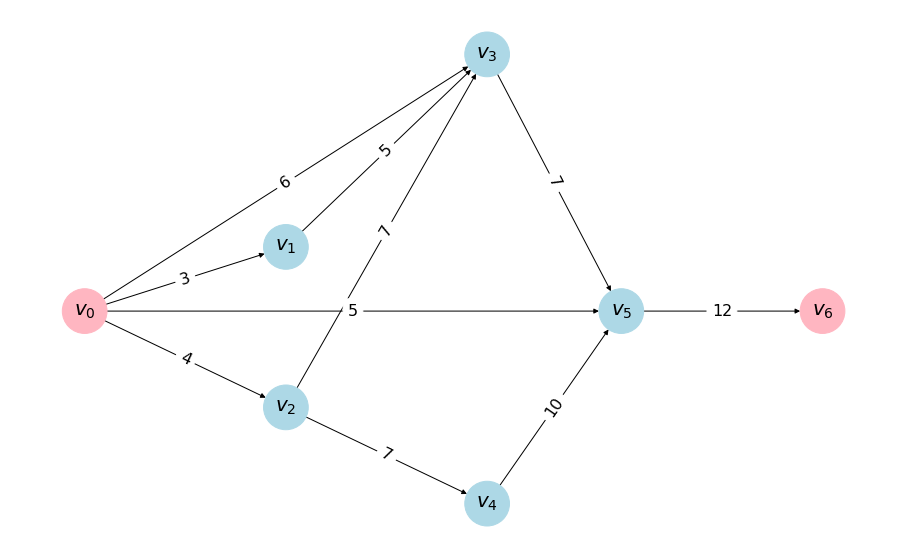}
	\caption{Example of a graph satisfying  \ref{eq:condition2_2}.}
	\label{fig_graph_special_case}
\end{figure}

To show how the path filters explained above  work for each of the algorithms, we will calculate the distance from node $v_0$ to $v_6$ in the graph illustrated in Figure \ref{fig_graph_special_case} (which satisfies the requirements given by \ref{eq:condition2_2}) using the weights associated with the path lengths $W:=\left(\frac{1}{t}\right)_{t=1}^\infty$. Since $v_6$ is only directly connected by $v_5$, the computation of $d_{\phi}(v_0,v_6)$ is determined by the paths $Q = P \sqcup (v_5,v_6)$, with $P \in \mathcal P(v_0,v_5)$.

\begin{itemize}
	\item[(1)] $P_1 = (v_0, v_1, v_3, v_5)$, with $l(P_1) = 3$, $s(P_1) = 15$ and $d(P_1) = 5$,
	\item[(2)] $P_2 = (v_0, v_2, v_3, v_5)$, with $l(P_2) = 3$, $s(P_2) = 18$ and $d(P_2) = 6$,
	\item[(3)] $P_3 = (v_0, v_2, v_4, v_5)$, with $l(P_3) = 3$, $s(P_3) = 21$ and $d(P_3) = 7$,
	\item[(4)] $P_4 = (v_0, v_3, v_5)$, with $l(P_4) = 2$, $s(P_4) = 13$ and $d(P_4) = 6.5$,
	\item[(5)] $P_5 = (v_0, v_5)$, with $l(P_5) = 1$, $s(P_5) = 5$ and $d(P_5) = 5$.
\end{itemize}

In this case, if we use the general path filtering given by  \ref{set:filter_path}, we would obtain possible paths
$$\mathcal P^{*}(v_0,v_5) = \{P_1, P_4, P_5\},$$
since $P_1 \preceq P_2 \preceq P_3$.

Instead, if we use the path filtering for this particular case given by \ref{set:filter_path2}, we have that $P_1 \preceq^1 P_2 \preceq^1 P_3$, $P_1 \preceq^1 P_4$ and $P_1 \preceq^1 P_5$ obtaining directly
$$
\mathcal P_1^{*}(v_0,v_5) = \{P_1\}.
$$

Therefore, for the calculation of $d_{\phi}(v_0,v_6)$ we only have to compute the distance associated to  the path $Q = P_1 \sqcup (v_5,v_6),$ that is, $d_{\phi}(v_0,v_6) = d(Q).$


\vspace{1cm}

\section*{Code Availability}
All the algorithms presented in this paper are available in a GitHub repository. It can be accessed at the following link:
\url{https://github.com/serjj99/path_length_weighted_distance}.

\section*{Acknowledgment}
	This research was funded by the Agencia Estatal de Investigaci\'on, grant number PID2022-138342NB-I00.
	The research of the first author was funded by the Universitat Polit\`ecnica de Val\`encia, Programa de Ayudas de Investigaci\'on y Desarrollo (PAID-01-21). The research of other authors was funded by the European Union’s Horizon Europe research and innovation program under the Grant Agreement No. 101059609 (Re-Livestock).


\begin{thebibliography}{999}
		\bibitem{bar83}
		Barnes, J.A.; Harary, F. Graph theory in network analysis. {\em Soc. Netw.} {\bf 1983}, {\em 5.2}, 235--244. 
		
		\bibitem{bell}
		Bellman, R. On a Routing Problem. {\em Q. Appl. Math.} \textbf{1958}, \emph{16(1)}, 87--90.
		
		\bibitem{bozzo13}
		Bozzo, E.; Franceschet, M. Resistance distance, closeness, and betweenness. {\em Soc. Netw.} {\bf 2013}, \emph{35(3)}, 460--469. 
		
		\bibitem{brandes}
		Brandes, U. \emph{Network analysis: methodological foundations}; Springer Science \& Business Media: Berlin, 2005.
		
		\bibitem{bu14}
		Bu, C.; Yan, B.; Zhou, X.; Zhou, J. Resistance distance in subdivision-vertex join and subdivision-edge join of graphs. {\em Linear Algebra Appl.} {\bf 2014}, {\em 458}, 454--462.
		
		\bibitem{buck}
		Buckley, F.; Harary, F. \emph{Distance in graphs}; Addison-Wesley:   Redwood City, CA, 1990.
		
		\bibitem{cal}
		Calabuig, J. M.; Falciani, H.; Sapena, A. F.; Raffi, L. G.; S\'anchez P\'erez, E.A. Graph distances for determining entities relationships: a topological approach to fraud detection. {\em IJITDM} {\bf 2023}, \emph{22(04)}, 1403--1438.
		
		\bibitem{che11}
		Chen, J.; Safro, I. Algebraic distance on graphs. {\em SISC} {\bf 2011}, \emph{33(6)}, 3468--3490.
		
		\bibitem{cheb11}
		Chebotarev, P. A class of graph-geodetic distances generalizing the shortest-path and the resistance distances. {\em Discrete Appl. Math.} {\bf 2011}, \emph{159(5)}, 295--302. 
		
		\bibitem{deza}
		Deza, M.M.; Deza, E. \emph{Encyclopedia of distances}; Springer: Berlin, Heidelberg, 2009.
		
		\bibitem{dij}
		Dijkstra, E.W. A note on two problems in connexion with graphs. \emph{Numer. Math. 1.} \textbf{1959}, \emph{1}, 269--271.
		
		\bibitem{entr76}
		Entringer, R.C.; Jackson, D.E.; Snyder, D.A. Distance in graphs. {\em  Czech. Math. J.} {\bf 1976}, \emph{26(2)}, 283--296.  
		
		\bibitem{fouss}
		Fouss, F.; Saerens, M.; Shimbo, M. {\em Algorithms and models for network data and link analysis}; Cambridge University Press: Cambridge, England, 2016.
		
		\bibitem{god11}
		Goddard, W.; Oellermann, O.R. Distance in graphs. {\em Structural Analysis of Complex Networks} \textbf{2011}, 49--72. 
		
		\bibitem{hak}
		Hakimi, S.L.; Yau, S.S. Distance matrix of a graph and its realizability. {\em Q. Appl. Math.} {\bf 1965}, \emph{22(4)}, 305--317.
		
		\bibitem{har53}
		Harary, F.; Norman, R.Z. Graph theory as a mathematical model in social science. {\em Ann Arbor: University of Michigan, Institute for Social Research, No. 2.} \textbf{1953}.
		
		\bibitem{kam89}
		Kamada, T.; Kawai, S. An algorithm for drawing general undirected graphs. {\em Inf. Process. Lett} {\bf 1989}, \emph{31(1)}, 7--15.
		
		\bibitem{klein02}
		Klein, D.J. Graph geometry via metrics. In {\em Topology in Chemistry} Woodhead Publishing: Sawston, UK, 2002; pp. 292--315.
		
		\bibitem{kle93}
		Klein, D.J.; Randi\'c, M. Resistance distance. {\em J. Math. Chem.} {\bf 1993}, \emph{12(1)}, 81--95.
		
		\bibitem{mer}
		Mester, A.; Pop, A.; Mursa, B.E.M.; Grebl\u{a}, H.; Diosan, L.; and Chira, C. Network analysis based on important node selection and community detection. \emph{Mathematics} \textbf{2021}, \emph{9(18)}, 2294.
		
		\bibitem{oer}
		Oehlers, M.; and Fabian, B. Graph metrics for network robustness—A survey. \emph{Mathematics} \textbf{2021}, \emph{9(8)}, 895.
		
		\bibitem{step}
		Stephenson, K.; Zelen, M. Rethinking centrality: Methods and examples. {\em Soc. Netw.} {\bf 1989}, \emph{11(1)}, 1--37.
		
		\bibitem{yan19}
		Yang, Y.; Klein, D.J. Two-point resistances and random walks on stellated regular graphs. {\em J. Phys. A-Math. Theor.} {\bf 2019}, \emph{52(7)}, 075201. 
		
		\bibitem{yang14}
		Yang, Y.; Klein, D.J. Comparison theorems on resistance distances and Kirchhoff indices of S, T-isomers. {\em Discrete Appl. Math.} {\bf 2014}, \emph{175}, 87--93.
		
	\end{thebibliography}
\end{document}